\documentclass[11pt]{article}
\usepackage{fullpage}
\usepackage{times}  % DO NOT CHANGE THIS
\usepackage{helvet}  % DO NOT CHANGE THIS
\usepackage{courier}  % DO NOT CHANGE THIS
\usepackage{graphicx} % DO NOT CHANGE THIS
\usepackage{caption} % DO NOT CHANGE THIS AND DO NOT ADD ANY OPTIONS TO IT
\usepackage{algorithm}
\usepackage{algorithmic}

\usepackage{newfloat}
\usepackage{listings}
\DeclareCaptionStyle{ruled}{labelfont=normalfont,labelsep=colon,strut=off} % DO NOT CHANGE THIS
\floatstyle{ruled}
\newfloat{listing}{tb}{lst}{}
\floatname{listing}{Listing}
\usepackage{amsmath}
\usepackage{amsthm}
\usepackage{amssymb}
\usepackage{epsfig}
\usepackage{tikz}
\newtheorem{theorem}{Theorem}[section]
\newtheorem{remark}{Remark}
\newtheorem{lemma}[theorem]{Lemma}
\newtheorem{claim}{Claim}[theorem]
\newtheorem{proposition}{Proposition}
\newtheorem{observation}[theorem]{Observation}

\theoremstyle{definition} 
\newtheorem{definition}[theorem]{Definition}

\newcommand{\ignore}[1]{}

\newcommand{\newfontobj}[2]{
  \newcommand{#1}[1]{
    \expandafter\def\csname##1\endcsname{{#2 ##1}}}}

\newfontobj{\class}{\rm}
\newfontobj{\lang}{\bf}

\class{NP}
\class{P}
\class{OptP}
\usepackage{amsmath,amsthm,amssymb, color, tikz}
\usepackage{tcolorbox}
\usepackage{thmtools} 
\usepackage{bbm}
\usepackage{dsfont}
\usetikzlibrary{positioning}
\usetikzlibrary{shapes}
\usetikzlibrary{decorations.pathreplacing}
\usepackage{enumerate, enumitem} 

\usepackage{mathtools}

\usepackage{hyperref}
\newcommand{\lexsat}{\ensuremath{\mathsf{LexSAT}}}
\newcommand{\optrust}{\ensuremath{\mathsf{optConf}}}
\newcommand{\maxsat}{\ensuremath{\mathsf{MaxSat}}}
\newcommand{\maxsatval}{\ensuremath{\mathsf{MaxSatVal}}}

\newcommand{\trust}{\ensuremath{\mathsf{Conf}}}
\newcommand{\optrustval}{\ensuremath{\mathsf{optConfVal}}}

\newcommand{\T}{\mbox{\rm T}}
\newcommand{\F}{\mbox{\rm F}}

\newcommand{\optsem}{\ensuremath{\mathsf{optSem}}}
\newcommand{\optsemval}{\ensuremath{\mathsf{optSemVal}}}

\newcommand{\optaccess}{\ensuremath{\mathsf{optAccess}}}
\newcommand{\optaccessval}{\ensuremath{\mathsf{optAccessVal}}}

\newcommand{\negation}{\mbox{$\daleth$}}
\newcommand{\sem}{\ensuremath{\mathsf{Sem}}}

\DeclareMathOperator*{\E}{\mathbb{E}}
\class{FP}
\title{Constraint Optimization over Semirings\thanks{The authors decided to forgo the old convention
of alphabetical ordering of authors in favor of a randomized ordering, denoted by \textcircled{r}. The publicly
verifiable record of the randomization is available at \href{https://www.aeaweb.org/journals/policies/
random-author-order/search}{https://www.aeaweb.org/journals/policies/
random-author-order/search}. An abridged version of the paper appeared in AAAI 2023.}
}

\usepackage{authblk}

\makeatletter
\newcommand\email[2][]%
  {\newaffiltrue\let\AB@blk@and\AB@pand
      \if\relax#1\relax\def\AB@note{\AB@thenote}\else\def\AB@note{\relax}%
        \setcounter{Maxaffil}{0}\fi
      \begingroup
        \let\protect\@unexpandable@protect
        \def\thanks{\protect\thanks}\def\footnote{\protect\footnote}%
        \@temptokena=\expandafter{\AB@authors}%
        {\def\\{\protect\\\protect\Affilfont}\xdef\AB@temp{\url{#2}}}%
         \xdef\AB@authors{\the\@temptokena\AB@las\AB@au@str
         \protect\\[\affilsep]\protect\Affilfont\AB@temp}%
         \gdef\AB@las{}\gdef\AB@au@str{}%
        {\def\\{, \ignorespaces}\xdef\AB@temp{\url{#2}}}%
        \@temptokena=\expandafter{\AB@affillist}%
        \xdef\AB@affillist{\the\@temptokena \AB@affilsep
          \AB@affilnote{}\protect\Affilfont\AB@temp}%
      \endgroup
      \let\AB@affilsep\AB@affilsepx
}
\makeatother

\renewcommand\Affilfont{\itshape\small}
\author{A. Pavan \textcircled{r}}
\affil{Iowa State Univerity}
\email{pavan@cs.iastate.edu}

\author{Kuldeep S.~Meel \textcircled{r}}
\affil{National University of Singapore, Singapore}
\email{meel@comp.nus.edu.sg}

\author{ N. V. Vinodchandran \textcircled{r}}
\affil{University of Nebraska-Lincoln}
\email{vinod@cse.unl.edu}
\author{Arnab Bhattacharyya}
\affil{National University of Singapore, Singapore}
\email{arnabb@nus.edu.sg}

\begin{document}
\maketitle

\begin{abstract}
    Interpretations of logical formulas over semirings (other than the Boolean semiring) have applications in various areas of computer science including logic, AI, databases, and security.  Such interpretations  provide richer information beyond the truth or falsity of a statement. Examples of such semirings include Viterbi semiring, min-max or access control semiring, tropical semiring, and fuzzy semiring. 
    
    The present work investigates the complexity of constraint optimization problems over semirings. The generic optimization problem we study is the following: Given a propositional formula $\varphi$ over $n$ variable and a semiring $(K,+,\cdot,0,1)$, find the maximum value over all possible interpretations of $\varphi$ over $K$. This can be seen as a generalization of the well-known satisfiability problem (a propositional formula is satisfiable if and only if the maximum value over all interpretations/assignments over the Boolean semiring is 1).  A related problem is to find an interpretation that achieves the maximum value. In this work, we first focus  on these optimization problems over the Viterbi semiring, which we call \optrustval\ and \optrust.   
    
    We first show that for general propositional formulas in negation normal form, \optrustval\ and {\optrust} are in ${\mathrm{FP}}^{\mathrm{NP}}$. We then investigate {\optrust} when the input formula $\varphi$ is represented in the conjunctive normal form.  For CNF formulae, we first derive an upper bound on the value of {\optrust} as a function of the number of maximum satisfiable clauses. In particular, we show that if $r$ is the maximum number of satisfiable clauses in a CNF formula with $m$ clauses, then its $\optrust$ value is at most $1/4^{m-r}$. Building on  this we establish that {\optrust} for CNF formulae is hard for the complexity class ${\mathrm{FP}}^{\mathrm{NP}[\log]}$. We also design polynomial-time approximation algorithms and establish an inapproximability for {\optrustval}. We establish similar complexity results for these optimization problems over other semirings including tropical, fuzzy, and access control semirings.
\end{abstract}

\section{Introduction}\label{sec:intro}

Classically, propositional formulae are interpreted over the Boolean semiring $\mathbb{B} = (\{\F,\T\},\vee,\wedge,{\F},{\T})$ which is the standard semantics for the logical truth. In this setting, the variables take one of the two values  \T\ (true) or \F\ (false).  However, it is natural to extend the semantics to other semirings. Here, the idea is to interpret logical formulae when the variables take values over a semiring $\mathbb{K}=(K, +, \cdot, 0, 1)$. Such interpretations  provide richer information beyond the truth or falsity of a statement and have applications in several areas such as databases, AI, logic, and security~(see \cite{IL89, FR97, Z97, CWW00, Cui02, GT17} and references therein). In particular, semiring {\em provenance analysis} has been successfully applied in several software systems, such as Orchestra and Propolis (see, e.g., \cite{ADT11, DMRT14, FGT08, G11, Tan13}).  

Examples of semirings that are studied in the literature include Viterbi semiring, fuzzy semiring, min-max or access control semiring, and tropical semiring. Semantics over the Viterbi semiring $\mathbb{V} = ([0,1],\max,\cdot,0,1)$  has 
applications in database provenance, where $x \in [0, 1]$ is interpreted as a {\em confidence score}~\cite{GT17, GKT07, Tan17, GradelM21}, in probabilistic parsing, in probabilistic CSPs, and in Hidden Markov Models~\cite{V67,KleinM03,BMR95}. The access control semiring can be used as a tool in security specifications~\cite{GT17}. Other semirings of interest include the {tropical semiring}, used in cost analysis and algebraic formulation for shortest path algorithms~\cite{Mohri02}, and fuzzy semirings used in the context of fuzzy CSPs~\cite{BMR95}.

Optimization problems over Boolean interpretations have been central in many application as well as foundation areas.
 Indeed, the classical satisfiability problem is determining whether a formula $\phi(x_1, \cdots, x_n)$ has an interpretation/assignment over the Boolean semiring that evaluates to True. Even though semiring semantics naturally appear in a variety of applications, the optimization problems over semirings, other than the Boolean semiring, have not received much attention.  

In this work, we introduce and investigate the complexity of  optimization problems over semiring semantics. Let $\mathbb{K}=({K}, +, \cdot, 0, 1)$ be a semiring with a total order over $K$ and $\varphi$ be a propositional formula over a set $X$ of variables.  A $\mathbb{K}$-interpretation  $\pi$ is a function from $X$ to $K$. Such an interpretation can be naturally extended to formula $\varphi$, which we denote by $\sem(\varphi, \pi)$.  We study the following computational problem:  Given  a propositional formula $\varphi$ in negation normal form over a set $X$ of variables,
compute the maximum value of $\sem(\varphi, \pi)$ over all possible interpretations $\pi$. We call this problem \optsemval. A related problem, denoted \optsem, is to compute an interpretation $\pi$ that maximizes $\sem(\varphi, \pi)$. Refer to Section~\ref{sec:prelims} for a precise formulation of these problems.

There has been a rich history of work which   formulated the notion of CSP over semirings and investigated local consistency algorithms in the general framework~\cite{B04,BG06,BMR95,BistarelliMR97,BMRS+99,MRS06}. These works did not involve interpretations and did not focus on the computational complexity of the above-defined problems. Relatedly, the computational complexity of  sum-of-product problems over semirings has been studied recently~\cite{EiterK21}. However, the  problems  they study are different from ours. To the best of our knowledge, optimization problems 
\optsem\ and \optsemval\ that we consider over semirings have not been studied earlier and there are no characterizations of their computational complexity. %As our results indicate that these problems are closely related to Boolean satisfiability problems in propositional logic.
\subsection{Our Results}
We comprehensively study the computational complexity of $\optsem$ and the related problem \optsemval\  over various semirings such as Viterbi semiring, tropical semiring, access control semiring and fuzzy semiring,  from both an algorithmic and a complexity-theoretic viewpoint. When the underlying semiring is the Viterbi semiring, we call these problems ${\optrust} $ and ${\optrustval}$.
Our results can be summarized as follows:
\begin{enumerate}[leftmargin=15pt]
    \item We establish that both  \optrust\ and \optrustval\ are  in  the complexity class $\mathrm{FP}^{\mathrm{NP}}$. The crucial underlying observation is that even though $\pi$ maps $X$ to real values in the range $[0,1]$; the solution to {\optrustval} can be represented using polynomially many bits. We then draw upon connections to Farey sequences to derive an algorithm with polynomially many $\mathrm{NP}$ calls (Theorem~\ref{theorem:pnp}).
    
    \item For CNF formulas, we establish an upper bound on  {\optrustval} as a function of the number of maximum satisfiable clauses (Theorem~\ref{thm:OptTrustBound}).
    
    \item We also establish a lower bound on the complexity of {\optrustval} and {\optrust}. In particular, we show that both the problems are hard for the complexity class $\mathrm{FP}^{\mathrm{NP}[\log]}$. To this end, we demonstrate a reduction from MaxSATVal to {\optrustval}; this reduction crucially relies on the above-mentioned upper bound on {\optrustval} in terms of the number of maximum satisfiable clauses (Theorem~\ref{thm:hardness}).
    
    \item We design a polynomial-time approximation algorithm for $\optrustval$ and establish an inapproximability result. In particular, for 3-CNF formulas with $m$ clauses,  we design a $0.716^m$-approximation algorithm and show that the approximation factor can not be improved to $0.845^m$ unless P = NP (Theorems~\ref{thm:appalgo} and~\ref{thm:inapprox}).
    \item Finally, we show that for the access control semiring, the complexity of these optimization problems is equivalent to the corresponding problems over Boolean semiring (Theorem~\ref{thm:accessequivalence}).
   % We also establish that there is no polynomial time algorithm for exponential approximation unless P=NP, which follows from the hardness of approximating MaxSAT. (Theorem~\ref{thm:inapprox}).
\end{enumerate}

\begin{remark}
Since Viterbi semiring and tropical semiring are isomorphic via the mapping $x \leftrightarrow -\ln x$, results established for Viterbi semiring also hold for the tropical semiring. Fuzzy semiring can be seen as an ``infinite refinement'' of access control semiring with the same algebraic structure, results that we establish for access control semiring  also hold for fuzzy semiring. 
\end{remark}
\noindent{\em Organization.} The rest of the paper is organized as follows. We give the necessary notation and definitions in Section~\ref{sec:prelims}. Section~\ref{sec:comptrust} details our results on the computational complexity of \optrust\ and \optrustval. Section~\ref{sec:approxtrust} deals with approximate algorithms and the hardness of approximation  of {\optrustval}. In Section~\ref{sec:access}, we give complexity results for optimization problems for the access control semiring.
Finally, we conclude in Section~\ref{sec:conclusion}.

\section{Preliminaries}\label{sec:prelims}

We assume that the reader is familiar with  definition of a semiring. We denote a generic semiring by $\mathbb{K}=(K,+,\cdot,0,1)$ where $K$ is the underlying set. For interpreting formulas over $\mathbb{K}$, we will add a ``negation'' function $\negation:K\rightarrow K$. 
We assume $\negation$ is a bijection so that  $\negation(\negation(x))=x$, and $\negation(0) = 1$. For ease of presentation, we use the most natural negation function (depending on the semiring). However, many of our results hold for very general interpretations of negation. Finally, as our focus is on optimization problems, we will also assume a (natural) total order on the elements of $K$.      

For a set $X = \{x_1, x_2, \ldots x_n\}$ of variables, we associate the set $\overline{X} = \{\neg x_1,\ldots, \neg x_n \}$. We call $X\cup \overline{X}$ the literals and formulas we consider are propositional formulas over $X\cup \overline{X}$ in {\em negation normal form}. We also view a propositional formula $\varphi$ in negation normal form as a rooted directed tree  wherein each leaf node is labeled with a literal, 1, or 0  and each internal node is labeled with conjunction $(\wedge)$ or disjunction $\vee$. Note that viewing $\varphi$ as a tree ensures  a similar size as its string representation.   We call the tree representing the formula $\varphi$ as {\em formula tree} and denote it with $T_{\varphi}$. 
For a propositional formula $\varphi(x_1, \cdots, x_n)$, in negation normal form we use $m$ to denote the size of the formula, i.e. the total number of occurrences of each variable and its negation. When $\varphi(x_1, \cdots x_n)$ is in CNF form, $m$ denotes the number of clauses. 
We interpret a propositional formula over a semiring $\mathbb{K}$ by mapping the variables to $K$ and naturally extending it. Formally, a $\mathbb{K}$-interpretation is a function $\pi: X \rightarrow K$. 
 We extend $\pi$ to an arbitrary propositional formula $\varphi$ in negation normal form, which is denoted by $\sem(\varphi,\pi)$ ($\sem$ stands for `semantics'), as follows. 
\begin{enumerate}
 \item[-] $\sem(x, \pi) = \pi(x)$ \label{rule:val-end}
    \item[-] $\sem(\neg x, \pi) = \negation(\pi(x))$ \label{rule:compl}
    \item[-] $\sem(\alpha \vee \beta,\pi) = \sem(\alpha,\pi) + \sem(\beta,\pi)$ \label{rule:val-begin} 
    \item[-] $\sem(\alpha \wedge \beta,\pi) = \sem(\alpha,\pi) \cdot \sem(\beta,\pi)$ 
  
\end{enumerate}

\subsection{Optimization Problems and Complexity Classes}
 For a formula $\varphi$, we define $\optsemval(\varphi)$ as  
\[ \optsemval(\varphi) = \max_{\pi} \{\sem(\varphi, \pi)\},\]
where $\max$ is taken over all possible $\mathbb{K}$-interpretations from $X$ to $K$.

\begin{definition}[\optsem\ and \optsemval] Given a propositional formula $\varphi$ in negation normal form, the \optsemval\ problem is to compute $\optsemval(\varphi)$. The \optsem\ problem is to compute a
$\mathbb{K}$-interpretation that achieves  $\optsemval(\varphi)$, i.e, output $\pi^*$ so that $\optsemval(\varphi) = \sem(\varphi,\pi^*)$.
\end{definition}

Notice that when $\mathbb{K}$ is the Boolean semiring (with $0 < 1$ ordering and standard negation interpretation), $\optsemval$ is the well-known satisfiability problem: the formula $\varphi$ is satisfiable if and only if $\optsemval(\varphi)=1$. Also, the problem $\optsem$ is to  output a satisfying assignment if the formula $\varphi$ is satisfiable. 

%\subsection{Examples of Semirings}
In this work, we consider the following semirings.

\begin{enumerate}

\item Viterbi semiring $\mathbb{V} = ([0,1],\max,\cdot,0,1)$. As mentioned, the Viterbi semiring has applications in database provenance, where $x \in [0, 1]$ is interpreted as confidence scores, in probabilistic parsing, in probabilistic CSPs, and in Hidden Markov Models.

\item The tropical semiring $\mathbb{T} = (\mathbb{R}\cup\{\infty\}, \min,+, \infty, 0)$.  

The tropical semiring is isomorphic to the Viterbi semiring via the mapping $x \leftrightarrow -\ln x$. 

\item The fuzzy semiring $\mathbb{F} = ([0, 1], \max, \min, 0, 1)$.
\item Access control semiring $\mathbb{A}_k = ([k], \max, \min, 0, k)$. Intuitively, each $i\in [k]$ is associated with an access control level with natural ordering. Here 0 corresponds to public access and $n$ corresponds to no access at all. $[k]$ is the set $\{0 < 1 < \cdots < k\}$.  
\end{enumerate}

Most of our focus will be on complexity of $\optsem$ and $\optsemval$ problems over the Viterbi semiring. We call the corresponding computational problems $\optrust$ and $\optrustval$ respectively. We call the extended interpretation function $\sem$ as $\trust$ in this case. 

\begin{definition}[\maxsat\ and \maxsatval]
Given a propositional formula $\varphi$ in CNF form, the \maxsat\  problem is to compute an assignment of $\varphi$ that satisfies the maximum number of clauses.  Given a propositional formula $\varphi$ in CNF form, the \maxsatval\  problem is to compute the  maximum number of clauses of $\varphi$ that can be satisfied. 
\end{definition}

We need a notion of reductions between functional problems. We use the notion of {\em metric reductions} introduced by Krentel~\cite{Krentel88}. 

\begin{definition}[Metric Reduction]
For two functions $f,g:\{0,1\}^* \rightarrow \{0, 1\}^*$, we say that $f$ metric reduces to $g$   
if there are polynomial-time computable functions $h_1$ and $h_2$ where $h_1:\{0,1\}^* \rightarrow \{0,1\}^*$ (the reduction function) and $h_2:\{0,1\}^* \times \{0,1\}^*\rightarrow \{0,1\}^*$ so that for any 
$x$, $f(x) = h_2(x,g(h_1(x)))$. 
\end{definition}

\begin{definition}
For a function $t:\mathbb{N} \rightarrow \mathbb{N}$, ${\rm FP}^{\NP[t(n)]}$ denotes the class of functions that can be solved in polynomial-time with $O(t(n))$ queries to an $\NP$ oracle where $n$ is the size of the input. When $t(n)$ is some polynomial, we denote the class by ${\rm FP}^{\NP}$. 
\end{definition}

Metric reductions are used to define notions of completeness and hardness for function classes  ${\rm FP}^{\NP}$ and ${\rm FP}^{\NP[\log]}$.
The following result due to Krentel~\cite{Krentel88} characterizes the complexity of the $\maxsatval$ problem.

\begin{theorem}[\cite{Krentel88}]
\maxsatval\ is complete for  ${\rm FP}^{\NP[\log]}$ under metric reductions. 
\end{theorem}

The following proposition is a basic ingredient in our results. It can be proved using basic calculus. 

\begin{proposition}\label{prop:basic}
Let $f(x) =  x^a(1-x)^b$ where $a,b$ are non-negative integers, the maximum value of $f(x)$ over the domain $[0, 1]$ is attained when $x = \frac{a}{a+b}$. The maximum value of the function is $\left(\frac{a}{a+b}\right)^a\left(\frac{b}{a+b}\right)^b$.

\end{proposition}

\section{Computational Complexity of Confidence Maximization}\label{sec:comptrust}

For semantics over Viterbi semiring we assume the standard closed world semantics and use the negation function $\negation(x) = 1-x$. Thus we have $\trust(\neg x,\pi) + \trust(x,\pi) = 1$. However, our upper bound proofs go through for any reasonable negation function. We discuss this  in Remark~\ref{rem:closed}. 

Since $\trust(\varphi,\pi)$ can be computed in polynomial time, $\optrust$ is at least as hard as $\optrustval$. The following observation states that computing $\optrustval$ and $\optrust$ are $\NP$-hard.
\begin{observation}\label{obs:sat}
For a  formula $\varphi$,  $\optrustval(\varphi) = 1$  if and only if $\varphi$ satisfiable. Hence
both \optrust\ and \optrustval\ are NP-hard. 
\end{observation}

While both $\optrust$ and $\optrustval$ are $\NP$-hard, we would like to understand their relation to other maximization problems. In the study of optimization problems, the complexity classes ${\rm FP}^{\NP}$ and ${\rm FP}^{\NP[\log]}$ play a key role. In this section, we investigate both upper and lower bounds for these problems in relation to the classes ${\rm FP}^{\NP}$ and ${\rm FP}^{\NP[\log]}$.

\paragraph{An Illustrative Example.} We first provide an illustrative example that gives an idea behind the upper bound. Consider the formula 
$\phi(x_1,x_2) = (x_1)\wedge (x_2)\wedge (\neg x_1 \vee \neg x_2)$. Clearly, the formula is not satisfiable. Over the Viterbi semiring the value of the $\optrustval = \max\limits_{x_i\in [0,1]}\left\{x_1x_2(1-x_1), x_1x_2(1-x_2)\right\}$ by distributivity. This is maximized when (by Proposition~\ref{prop:basic})  $x_1=1$ and $x_2=0.5$ or $x_1=0.5$ and $x_2=1$, leading to an optimum value of $0.25$. 
In the following section, we show that the computation of {\optrustval} reduces to maximization over a set of polynomial terms wherein each polynomial term  corresponds to a {\em proof tree}, which we define. While the number of polynomial terms could be exponential, we use an NP oracle to binary search for the term that gives the maximum value. 

\subsection{An Upper Bound for General Formulae}\label{subsec:general}
We show that $\optrustval$ and $\optrust$  can be computed in polynomial-time with oracle queries to an $\NP$ language. 

\begin{theorem}\label{theorem:pnp}
\optrustval\ for formulas in negation normal form is in ${\rm FP}^{\NP}$. 
\end{theorem}

\noindent{\em Proof Idea:} In order to show that \optrustval\ is in ${\rm FP}^{\NP}$, we use a binary search strategy using a language in $\NP$. One of the challenges is that the confidence value could potentially be any real number in $[0,1]$ and thus apriori we may not be able to bound the number of binary search queries.  However, we first argue  that for any formula $\varphi$ on $n$ variables and with size $m$, $\optrust(\varphi)$ 
is a fraction of the form $A/B$ where $1\leq A \leq B\leq 2^{nm\log m}$.  Ordered  fractions of such form are known as {\em Farey} sequence of order $2^{nm\log m}$ (denoted as ${\mathcal F}_{2^{nm\log m}}$). Thus our task is to do a binary search over ${\mathcal F}_{2^{nm\log m}}$ with time complexity $O(nm\log m)$. However, in general binary search for an unknown element in the Farey sequence ${\mathcal F}_N$ with time complexity $O(\log N)$ appears to be unknown. We overcome this difficulty by using an $\NP$ oracle to aid the binary search. We will give the details now.      

\begin{definition}
Let $\varphi(x_1,\cdots,x_n)$ be a propositional formula in negation normal form with size $m$. Let $T_\varphi$ be its formula tree. A proof tree $T$ of $T_\varphi$ is a subtree obtained by the following process: for every OR node $v$, choose one of the sub-trees of $v$. For every AND node $v$, keep all the subtrees.   
\end{definition}

Note that in a proof tree every OR node has only one child. 

\begin{definition}
Let $\varphi(x_1,\cdots,x_n)$ be a propositional formula in negation normal form and let  $T$ be a proof tree. We define the {\em proof tree polynomial} $p_T$ by inductively defining a polynomial for the subtree at every node $v$ (denoted by $p_v$): If the node $v$ is a variable $x_i$, the polynoimal is $x_i$ and if it is $\neg x_i$, the polynomial is $(1-x_i)$.  If $v$ is an AND node with children $v_1,\ldots,v_s$, then $p_v = \prod_{i=1}^s p_s$. If $v$ is an OR node with a child $u$, then $p_v = p_u$.       
\end{definition}

\begin{claim}\label{clm:prooftree}
Let $\varphi(x_1,\cdots,x_n)$ be a propositional formula in negation normal form and let $T$ be a proof tree of $\varphi$. 
\begin{enumerate}

\item The proof tree polynomial $p_T$ is of the form
\begin{equation*}%\label{eqn:exp}
\prod_{i=1}^n x_i^{a_i}(1-x_i)^{b_i}
\end{equation*}

where $0\leq a_i+b_i \leq m.$
\item For a $\mathbb{V}$-interpretation $\pi$,
\[\trust(T,\pi) = p_T\left(\pi(x_1),\ldots,\pi(x_n)\right ).\]

\item Both $\optrust(T)$ and $\optrustval(T)$ can be computed in polynomial-time.

\item\label{eqn:valT} $\optrustval(T) = \Pi_{i=1}^n \left(\frac{a_i}{a_i+b_i}\right)^{a_i}\left(\frac{b_i}{a_i+b_i}\right)^{b_i}.$

\end{enumerate}
\end{claim}

\begin{proof}
Item (1) follows from the definition of the proof tree polynomial and a routine induction and the fact that the size of the formula $\varphi$ is $m$.  Item (2) follows from the definitions. 

Note that the polynomial $\pi_{i=1}^n x_i^{a_i} (1-x_i)^{b_i}$ can be maximized by maximizing each of the individual terms $x_i^{a_i}(1-x_i)^{b_i}$. By Proposition~\ref{prop:basic}, the maximum value for a polynomial of this form is achieved at $x_i = \frac{a_i}{a_i+b_i}$. Thus the interpretation $\pi(x_i) = \frac{a_i}{a_i+b_i}$ is an optimal $\mathbb{V}$-interpretation that can be computed in polynomial-time. Since $0 \leq a_i+b_i \leq m$, $\optrustval$ also can be computed in polynomial-time.
Item~(4)  follows from Item~(3), by substituting the values $\pi(x_i)$ for in the polynomial $p_T$. 
\end{proof}

The next claim relates $\optrust$ of the formula $\varphi$ to $\optrust$ of its proof trees.  The proof of this claim follows from the definition of proof tree and standard induction.
\begin{claim}\label{clm:maxtree}
For a formula $\varphi$, 
\[
\optrustval(\varphi) = \max_{T}\optrustval(T)
\]
where maximum is taken over all proof trees $T$ of $T_\varphi$. 
If $T^*$ is the proof tree for which $\optrust(T)$ is maximized, then $\optrust(T^*) = \optrust(\varphi)$.
\end{claim}

The above claim states that $\optrust(\varphi)$ can be computed by cycling through all proof trees $T$ of $\varphi$ and computing $\optrust(T)$. Since there could be exponentially many proof trees, this process would take exponential time. Our task is to show that this process can be done in ${\rm FP}^{\NP}$. To do this we establish a claim that restricts values that 
$\optrustval(\varphi)$ can take. We need the notion of {\em Farey sequence}. 

\begin{definition}
For any positive integer $N$, the {\em Farey sequence} of order $N$, denoted by ${\mathcal F}_N$, is the set of all irreducible fractions $p/q$  with $0< p < q \leq N$ arranged in increasing order.
\end{definition}

\begin{claim}\label{claim:abyb}
 
\begin{enumerate}
\item  For a propositional formula $\varphi(x_1, \cdots, x_n)$, $ \optrustval(\varphi)$ belongs to the Farey sequence ${\mathcal F}_{2^{nm\log m}}$.

%is of the form $\frac{A}{B}$, where $1 \leq A, B, \leq 2^{nm\log m}$. 
\item For any two fractions $u$ and $v$ from ${\mathcal F}_{2^{nm\log m}}$, 
$|u - v| \geq 1/ 2^{2nm\log m}$
\end{enumerate}
\end{claim}

\begin{proof}
By Claim~\ref{clm:maxtree}, $\optrustval(\varphi)$ equals $\optrustval(T)$, for some proof tree $T$. By Item~(4) of Claim~\ref{clm:prooftree} this value is a product of fractions, where the denominator of each fraction is of the form $(a_i+b_i)^{a_i+b_i}$ where  $a_i$ and $b_i$ are non-negative integers. Since $a_i+b_i \leq m$, each denominator is at most $m^m$, and thus the denominator of the product is bounded by $m^{nm} = 2^{nm \log m}$.  Since the numerator is at most the denominator, the claim follows.

For the proof of the second part, let $u = p_1/q_1$ and $v = p_2/q_2$, $u > v$. Now $u - v = (p_1q2-p_2q_1)/q_1q_2$. Since $q_1, q_2 \leq 2^{nm \log m}$, we have $u - v > p_1q_2-p_2q_1/2^{2nm\log m}$. Since $p_1, p_2, q_1$, $q_2$ are all integers, $p_1q_2 - p_2q_1 \geq 1$. Thus $|u -v | \geq 1/2^{2nm\log m}$.

\end{proof}

Consider the following language
\[L_{\it opt} = \{\langle \varphi, v\rangle~|~ {\rm \optrustval}(\varphi) \geq v\}\]

\begin{claim}\label{clm:loptnp}
$L_{\it opt}$ is in $\NP$.
\end{claim}

\begin{proof}
Consider the following non-deterministic machine $M$. On input $\varphi$, $M$ guesses a proof tree $T$ of $\varphi$: for every OR node, non-deterministically pick  one of the subtrees. For $T$, compute $\optrustval(T)$ and accept if $\optrustval(T) \geq v$. This can be done in polynomial-time  using  Item~(3) of Claim \ref{clm:prooftree}.  The correctness of this algorithm follows from Claim~\ref{clm:maxtree}. 
\end{proof}

We need a method that given two fractions $u$ and $v$ and an integer $N$, outputs a fraction $p/q: u\leq p/q \leq v$, and $p/q \in {\mathcal F}_N$.  We give an ${\rm FP}^{\NP}$ algorithm that makes $O(N)$ queries to the $\NP$ oracle to achieve this.  We first define the $\NP$ language $L_{\it farey}$. For this we fix any standard encoding of fraction using the binary alphabet. Such an encoding will have $O(\log N)$ bit representation for any fraction in ${\mathcal F}_N$.  
\[ L_{\it farey} = \{ \langle N,u,v, z \rangle \mid \exists z'; u \leq zz' \leq v ~\&~  zz' \in {\mathcal F}_N  \} \]

The following claim is easy to see. 

\begin{claim}
$L_{\it farey} \in \NP$.
\end{claim}

Now we are ready to prove the Theorem~\ref{theorem:pnp}.

\begin{proof} ({\em of Theorem~\ref{theorem:pnp}}). 
The algorithm  performs a binary search over the range $[0, 1]$ by making adaptive queries $\langle \varphi,v\rangle$ 
to the $\NP$ language $L_{\it opt}$ starting with $v=1$. At any iteration of the binary search, we have an interval $I=[I_l, I_r]$ and with the invariant  $I_l \leq {\rm \optrustval}(\varphi) < I_r$. The binary search stops when the size of the interval $[I_l,I_r] =1/2^{2nm\log m}$. Since each iteration of the binary search reduces the size of the interval by a factor of 2, the search stops after making $2nm\log m$ queries to $L_{\it opt}$. The invariant ensures that $\optrustval(\varphi)$ is in this interval. Moreover, $\optrustval(\varphi) \in {\mathcal F}_{2^{nm\log m}}$ (by item (1) of Claim~\ref{claim:abyb}) and there are no other fractions from ${\mathcal F}_{2^{nm\log m}}$ in this interval (by item (2) of Claim~\ref{claim:abyb}).  Now, by making $O(nm\log m)$ queries to $L_{\it farey}$ with $N=2^{nm\log m}$, $u=I_l$, $v=I_r$, we can construct the binary representation of the unique fraction in ${\mathcal F}_{2^{nm\log m}}$ that lies between  $I_l$ and $I_r$ which is $\optrustval(\varphi)$.  
\end{proof}

Next we show the optimal $\mathbb{V}$-interpretation can also be computed in polynomial time with queries to an \NP\ oracle. 

\begin{theorem}\label{thm:pnpcompute}
\optrust\ for formulas in negation normal form can be computed in ${\rm FP}^{\NP}$.   
\end{theorem}
\begin{proof}
Let $\varphi$ be a propositional formula in negation normal form. We use a prefix search over the encoding of proof  trees of $\varphi$ using an $\NP$ language to isolate a proof tree $T$ such that $\optrustval(\varphi) = \optrustval(T)$. For this, we fix an encoding of proof trees of $\varphi$. Consider the following $\NP$ language $L_{\it pt}$:
\begin{align*}
\{\langle \varphi, v, z\rangle \mid \exists z': &zz' \mbox{encodes a proof tree $T$ of $\varphi$}\\
&~\&~ \optrustval(T) = v \}
\end{align*}

\begin{claim}\label{clm:lptnp}
$L_{\it pt}$ is in \NP. 
\end{claim}

\begin{proof}
Consider a non-deterministic machine that guesses a $z'$, verifies that $zz'$ encodes a proof tree $T$ of $\varphi$, and accepts if $\optrustval(T) = v$. By item (3) of Claim~\ref{clm:prooftree}, $\optrustval(T)$ can be computed in polynomial time.  
\end{proof}

To complete the proof Theorem \ref{thm:pnpcompute}, given a propositional formula $\varphi$, we first use ${\rm FP}^{\NP}$ algorithm from Theorem~\ref{theorem:pnp} to compute $v^* = \optrustval(\varphi)$. Now we can construct a proof tree $T$ of $\varphi$ so that $\optrustval(T) = v^*$ by a  prefix search using language $L_{\it pt}$.  
Now by Claim~\ref{clm:prooftree}, we can compute a $\mathbb{V}$-interpretation $\pi^*$ so that $\trust(T,\pi^*) = v^*$. Thus $\pi^*$ is an optimal $\mathbb{V}$-interpretation for $\varphi$, by Claim~\ref{clm:maxtree}.
\end{proof}

%%%%%%%%% Upper Bound for DNF Omitted
\ignore{
\subsubsection*{Upper Bound for DNF formulae} 

We note that for formulae presented in DNF form, both $\optrust$ and $\optrustval$ has polynomial-time algorithms. 
\begin{theorem}
Both $\optrust$ and $\optrustval$ can be computed in polynomial-time when the input formula is in DNF form.
\end{theorem}

\begin{proof}
If $\varphi$ is a DNF formula with $m$ terms, then the number of proof trees is exactly $m$, where each term of the formulas is a proof tree. The algorithms cycles through all proof trees $T$ of $\varphi$ and computes $\optrust(T)$ and $\optrustval(T)$ and outputs the $\optrustval(T)$ and an interpretation that achieves the maximum. The correctness follows from Claim~\ref{clm:maxtree}.
\end{proof}
}
%%%%%%%%%%%%%%%%%%%%%%%%%%%%%%%%%%%%%%%

\begin{remark}
\label{rem:closed}
We revisit the semantics of negation.  As stated earlier, by assuming the closed world semantics, we have $\negation(x) = 1-x$. We note that this assumption is not strictly necessary for the above proof to go through. Recall that Item (1) of Claim~\ref{clm:prooftree}  states that the proof tree polynomial is of the form $\prod x_i^{a_i}(1-x_i)^{b_i}$.
For a general negation function $\negation$, the proof tree polynomial is of the form $\prod x_i^{a_i}(\negation(x_i))^{b_i}$.  Now if the maximum value of a term $x^a(\negation(x))^b$  can be found, for example when $\negation$ is an explicit differentiable function, the result will hold.

\end{remark}

\subsection{Relation to $\maxsat$ for CNF Formulae}\label{subsec:maxsat}

In this section we study the $\optrustval$ problem for CNF formulae and establish its relation to the $\maxsat$ problem. We first exhibit an upperbound on the $\optrustval(\varphi)$ using the maximum number of satisfiable clauses.  Building on this result, in Section~\ref{sec:hardness} we show that $\optrustval$ for CNF formulae is hard for the complexity class $\FP^{\NP[\log]}$. 

We first define some notation that will be used in this  and next subsections.  Let $\varphi(x_1, \cdots x_n) = C_1 \wedge \cdots \wedge C_m$ be a CNF formula and let $\pi^*$ be an optimal $\mathbb{V}$-interpretation.  For each clause $C$ from $\varphi$, let $\pi^*(C)$ be the value achieved by this interpretation, i.e $\pi^*(C) = \trust(C, \pi^*)$. Observe that since $C$ is a disjunction of literals, $\pi^*(C) = \max_{\ell \in C}\{\pi^*(\ell)\}$. For a clause $C$, let
\[\ell_C = {\rm argmax}_{\ell \in C} \{\pi^*(\ell)\}\]

In the above, if there are multiple maximums, we take the smallest literal as $\ell_C$ (By assuming an order $x_1 < \neg{x_1} < x_2 < \neg{x_2} \cdots < x_n < \neg{x_n})$.  Observe  that, since we are working over the Viterbi semiring,  $\trust(C, \pi^*) = \pi^*(\ell_C)$.  A literal $\ell$ is {\em maximizing literal for a  clause $C$}, if $\ell_C = \ell$. 

Since $\varphi$ is a CNF formula, for any $\mathbb{V}$-interpretation $\pi$  $\trust(\varphi, \pi)$ is of the form 
$\Pi_{i=1}^m \trust(C_i, \pi)$.
Given a collection of clauses ${\mathcal D}$ from $\varphi$, the {\em contribution of ${\mathcal D}$ to $\trust(\varphi, \pi)$} is defined as $\Pi_{c \in {\mathcal D}} \trust(C, \pi)$.

The following theorem provides an upperbound on $\optrustval(\varphi)$ using $\maxsatval$. This is the main result of this section.

\begin{theorem}\label{thm:OptTrustBound}
Let $\varphi(x_1, \cdots, x_n)$ be a CNF formula with $m$ clauses. Let $r$ be the maximum number of clauses that can be satisfied.  Then  $\optrustval(\varphi) \leq 1/4^{(m-r)}$.
\end{theorem}

\begin{proof}

Let $\pi^*$ be an optimal $\mathbb{V}$-interpretation for $\varphi$.
A clause  $C$ is called {\em low-clause} if $ \pi^*(C) < 1/2$,  $C$ is called a {\em high-clause} of $\pi^*(C) > 1/2$, and  $C$ is a {\em neutral-clause} if $\pi^*(C) = 1/2$. Let $L$, $H$, and $N$ respectively denote the number of low, high, and neutral clauses.

We start with the following claim that relates the number of neutral clauses and the number of high-clauses to $r$.

\begin{claim}\label{clm:NeutralHigh}
$\frac{N}{2}+ H \leq r$
\end{claim}

\begin{proof}
Suppose that the number of low-clauses is strictly less than $m-r$,  thus number of high-clauses is more than $r$. 

For a variable $x$, let
\[p_x = |\{C~|~\mbox{$C$ is neutral and } \ell_C = x\}|\]
and
\[q_x =|\{C~|~\mbox{$C$ is neutral and } \ell_C = \neg{x}\}|\]

That is $p_x$ is the number of neutral clauses for which $x$ is the maximizing literal and $q_x$ is the number of neutral clauses for which $\neg x$ is the maximizing literal. 

Consider the  truth assignment that is constructed  based on the following three rules:  (1) For every high-clause $C$, set $\ell_C$ to True and  $\neg{\ell_C}$ to False, 2) For every variable $x$, if one of $p_{x}$ or $q_x$ is not zero, then if $p_x \geq q_x$, then set $x$ to True, otherwise set $x$ to False. (3) All remaining variables are consistently assigned arbitrary to True/False values.

We argue that this is a consistent assignment: I.e, for every literal $\ell$, both $\ell$ and $\neg{\ell}$ are not assigned the same truth value.
Consider a literal $\ell$. If there is a high clause $C$ such that $\ell = \ell_C$, then this literal is assigned truth value True and $\neg{\ell}$ is assigned False. In this case, since $\pi^*(\ell) > 1/2$,  $\pi^*(\neg{\ell}) < 1/2$. Thus $\neg{\ell}$ can not be maximizing literal for any high-clause and thus Rule (1) does not assign True to $\neg{\ell}$. Again, since $\pi^*(\ell) > 1/2$,  there is no neutral-clause  $D$ such that $\ell = \ell_D$ or $\neg{\ell} = \ell_D$.  Thus Rule (2) does not assign a truth value to either of $\ell$ or $\neg{\ell}$.  Since $\ell$ and $\neg{\ell}$ are assigned truth values, Rule (3) does not assign a truth value to $\ell$ or $\neg{\ell}$. 

Consider a variable $x$ where  at least one of $p_x$ or $q_x$ is not zero. In this case $x$ or $\neg{x}$ is maximizing literal for a neutral clause. Thus $\pi^*(x) = \pi^*(\neg{x}) = 1/2$ and neither $x$ nor ${\neg{x}}$ is maximizing literal for a high-clause. Thus Rule (1) does not assign a truth value to  $x$ or ${\neg{x}}$.  Now $x$ is True if and only if $p_x \geq q_x$, thus the truth value assigned to $x$ (and $\neg{x}$) is consistent. Since Rule (3) consistently assigns truth values of literals that are not covered by Rules (1) and (2), the constructed assignment is a consistent assignment.

For every high clause $C$, literal $\ell_C$ is set to true. Thus the assignment satisfies all the high-clauses. Consider a variable $x$ and let $\mathcal{D}$ be the (non-empty) collection of neutral clauses for which either $x$ or $\neg{x}$ is a maximizing literal. As $x$ is assigned True if and only if $p_x \geq q_x$, at least half the clauses from $\mathcal{D}$ are satisfied.  Thus this assignment satisfies at least $H + \frac{N}{2}$ clauses. Since $r$ is the maximum number of satisfiable clauses, the claim follows.
\end{proof}

For a literal $\ell$, let $a_\ell$ be the number of low-clauses $C$ for which $\ell$ is a maximizing literal, i.e,
\[a_{\ell} = |\{C~|~\mbox{$C$ is a low-clause and } \ell_C = \ell\}|,\]
and
\[b_{\ell} = |\{C~|~\mbox{$C$ is a high-clause and } \ell_C = \neg{\ell}\}|,\]

We show the following relation between $a_\ell$ and $b_\ell$. 

\begin{claim}\label{clm:highlow}
For every literal $\ell$, $a_\ell \leq b_\ell$.
\end{claim}

\begin{proof}

\begin{equation}\label{eqn:subformula}
\trust(\varphi, \pi) = \Pi_{i} \trust(\varphi_{\mid x_i}, \pi) 
\end{equation}

Now suppose that $a_\ell > b_\ell$ for some literal $\ell$. Let $x_j$ be the variable corresponding to the literal $\ell$.
Note that 
\[\trust(\varphi_{\mid x_j}, \pi^*) = \pi^*(\ell)^{a_\ell} \times (1 - \pi^*(\ell))^{b_\ell}\]

where $\pi(\ell) < 1/2$. Consider a new interpretation $\pi'$ where $\pi'(\ell) = 1 - \pi^*(\ell)$, and for all other literals the value of $\pi'$ is the same as the value of $\pi^*$. Now

\begin{eqnarray*}
{\trust(\varphi_{\mid x_j}, \pi') \over \trust(\varphi_{\mid x_j}, \pi^*) } & = &{\pi'(\ell)^{a_\ell} \times (1- \pi'(\ell))^{b_{\ell}} \over \pi(\ell)^{a_\ell} \times (1 - \pi(\ell))^{b_\ell}} \\
& = & {(1-\pi(\ell))^{a_\ell} \times \pi(\ell)^{b_\ell} \over \pi(\ell)^{a_\ell} \times (1 - \pi(\ell))^{b_\ell}}\\
& = &\left( {(1- \pi(\ell) \over \pi(\ell)} \right)^{a_\ell - b_\ell}
 >  1
\end{eqnarray*}

The  last inequality follows because $\pi(\ell) < 1/2$ and the assumption that $a_\ell > b_\ell$.  Since $\trust(\varphi_{\mid x}, \pi^*) = \trust(\varphi_{\mid x}, \pi')$ for every $x \neq x_j$, combining the above inequality with Equation~\ref{eqn:subformula}, we obtain that $\trust(\varphi, \pi') > \trust(\varphi, \pi^*)$ and thus $\pi^*$ is not an optimal $\mathbb{V}$-interpretation.  This is a contradiction. Thus $a_\ell \leq b_\ell$
\end{proof}

We next bound the contribution of neutral and low clauses to $\optrustval(\varphi)$. For every neutral clause $C$, $\pi^*(C) = 1/2$, thus we have the following observation.
%Equipped with above claims,  we complete the proof of the theorem.

\begin{observation}\label{obs:neutral}
The contribution of neutral clauses to $\trust(\varphi, \pi^*)$ is exactly $1/2^N$. 
\end{observation}

We establish the following claim. 

\begin{claim}\label{clm:bound}
\[\trust(\varphi, \pi^*) =\prod_{\ell} \left (\pi^*(\ell)^{a_\ell} \times  (1- \pi^*(\ell))^{b_\ell}\right) \times \frac{1}{2^N}\]

\end{claim}

\begin{proof}
By Observation~\ref{obs:neutral}, the contribution of neutral clauses to $\trust(\varphi, \pi^*)$ is $1/2^N$. Next we show that the contribution of all high and low clauses is exactly. 
\[ \prod_{\ell} \pi^*(\ell)^{a_{\ell}} \times (1-\pi^*(\ell))^{b_{\ell}}.\]

For this we first claim that exactly one of $\ell$ or $\neg{\ell}$ contribute to the above product.  For this it suffices to prove that for every literal $\ell$ exactly one of $a_{\ell}$ ($b_{\ell}$ resp.) or $a_{\neg{\ell}}$ ($b_{\neg{\ell}}$) is zero. Suppose $a_\ell \neq 0$, in this case $\neg{\ell}$ can not be maximizing literal for any low clause. Thus $a_{\neg{\ell}} = 0$. Suppose that $b_{\ell} \neq 0$, then $\neg{\ell}$ is a maximizing literal for a high clause and thus $\pi^*(\neg{\ell}) > 1/2$, and $\pi^*(\ell) \leq 1/2$. If $b_{\neg{\ell}} \neq 0$, then $\ell$ must be a maximizing literal for a high-clause, and this is not possible as $\pi^*(\ell) \leq 1/2$. Thus $b_{\neg{\ell}} = 0$.

Let $Z$ be the collection of literals $\ell$ for which $a_{\ell} > 0$. Now that quantity $\prod_{\ell \in Z} \pi^*(\ell)^{a_{\ell}} \times (1- \pi^*(\ell))^{b_{\ell}}$ captures the contribution of all low clauses and $\sum_{\ell \in Z}$ many high-clauses. For all remaining high-clauses, there exist a literal $\ell$ such that $\ell \notin Z$ and $b_{\ell} \neq 0$. The contribution of all the remaining high- clauses is $\prod_{\ell \notin Z} (1- \pi(\ell))^{b_{\ell}}$. This quantity equals $\prod_{\ell \notin Z} \pi^*(\ell)^{a_{\ell)}} \times (1-\pi(\ell))^{b_{\ell}}$ as $a_{\ell} = 0$ for $\ell \notin Z$.
\end{proof}

Finally, we are ready to complete the proof of Theorem~\ref{thm:OptTrustBound}. For every literal $\ell$, By Claim~\ref{clm:highlow}, $a_\ell \leq b_\ell$. Let $b_{\ell}  =a_{\ell} + c_{\ell}$, $c_\ell \geq 0$.
Consider the following
inequalities.
\begin{eqnarray*}
\optrustval(\varphi) & = & \trust(\varphi, \pi^*)\\
& = & \prod_{\ell} \left(\pi^*(\ell)^{a_{\ell}} \times (1-\pi^*(\ell))^{b_{\ell}} \right)\times \frac{1}{2^N} \\
& =& \prod_{\ell} \left(\pi^*(\ell)^{a_\ell} \times (1-\pi^*(\ell))^{a_{\ell}+c_{\ell}}\right) \times \frac{1}{2^N}\\
& \leq & \prod_{\ell} \left(\pi^*(\ell)^{a_\ell} \times (1-\pi^*(\ell))^{a_{\ell}} \right) \times \frac{1}{2^N}\\
& \leq & \prod_{\ell} \left(\frac{1}{4^{a_{\ell}}}\right) \times \frac{1}{2^N} = \frac{1}{4^{L+N/2}} \\
%& = & \frac{1}{4^{L}} \times \frac{1}{2^N} = \frac{1}{4^{L+N/2}}
\end{eqnarray*}
In the above, equality at line 2 is due to Claim~\ref{clm:bound}. The  inequality at line 4 follows because $(1- \pi^*(\ell)) \leq 1$. The last inequality follows because $x(1-x)$ is maximized at $x = 1/2$. The last equality follows  as $\sum a_\ell = L$.
Note that the number of clauses $m = N + H + L$ and by Claim~\ref{clm:NeutralHigh} $H+N/2 \leq r$. 
It follows that $L + N/2 \geq m-r$. Thus $\optrustval(\varphi) = \trust(\varphi,\pi^*) \leq {1\over 4^{L+N/2}} \leq {1\over 4^{m-r}}$.
\end{proof}

\subsection{$\FP^{\NP[\log]}$- Hardness}\label{sec:hardness}
 %Can we attain $1/4^{m-k}$ for every formula? or are there formulae where this cannot be attained? 

In this subsection, we show that $\optrustval$ is hard for the class ${\rm FP}^{\NP[\log]}$. We show this by reducing \maxsatval\ to \optrustval. Since \maxsatval\ is complete for ${\rm FP}^{\NP[\log]}$, the result follows. We also show that the same reduction can be used to compute a \maxsat\ assignment from an optimal $\mathbb{V}$-interpretation. 

\begin{theorem}\label{thm:hardness}
\maxsatval\ metric reduces to \optrustval\  for CNF formulae. Hence $\optrustval$ is hard for ${\mathrm{FP}}^{\mathrm{NP}[\log]}$ for CNF formulae.
\end{theorem}
\begin{proof}
Let $\varphi(x_i,\ldots,x_n) = C_1 \wedge \ldots \wedge C_m$ be a formula with $m$
 clauses on variables $x_1,\ldots,x_n$. Consider the formula $\varphi'$ with $m$ additional variables $y_1,\ldots,y_m$ constructed as follows: For each clause $C_i$  of $\varphi$, add the clause $C'_i = C_i\vee y_i$ in $\varphi'$. Also add  $m$ unit clauses $\neg y_i$. That is
 \[
 \varphi' = (C_1 \vee y_1)\wedge \ldots \wedge (C_m \vee y_m) \wedge \neg y_1 \wedge \cdots \wedge \neg y_m
 \]
 
 \begin{claim}\label{clm:earlierargument}
  $\optrustval(\varphi') = \frac{1}{4^{m-r}}$ where $r$ is the maximum number of clauses that can be satisfied in $\varphi$. 
 \end{claim} 
 \begin{proof} 
 We show this claim by first showing that $\optrustval(\varphi') \leq \frac{1}{4^{m-r}}$ and exhibiting an interpretation $\pi^*$ so that $\trust(\varphi,\pi^*) = \frac{1}{4^{m-r}}$. 
 We claim that if  $r$ is the maximum number of clauses that can be satisfied in $\varphi$, then $m+r$ is the maximum number of clauses that can be satisfied in $\varphi'$. We will argue this by contradiction. Let $\mathbf{a}$  be an assignment that satisfies  $> m+r$ clause in $\varphi'$. Let $s$ be the number of $y_i$s that are set to False. This assignment will satisfy  $m-s$ clauses of the form $C_i \vee y_i$. However the total number of clauses of the form $C_i \vee y_i$ that are satisfied is $> m+r -s$. Thus there are $> r$ clauses of the form $C_i \vee y_i$ that are satisfied where $y_i$ is set to False. This assignment when restricted to $x_i$s will satisfy more than $r$ clauses of $\varphi$. Hence the contradiction.    
 
 Thus from Theorem~\ref{thm:OptTrustBound}, it follows that $\optrustval(\varphi') \leq \frac{1}{4^{m-r}}$.   Now we exhibit an interpretation $\pi^*$ so that $\trust(\varphi, \pi^*) = \frac{1}{4^{m-r}}$.  Consider an assignment $\mathbf{a} = a_1,\ldots, a_n$ for $\varphi$ that satisfies $r$ clauses.  Consider the following interpretation $\pi^*$ over the variable of $\varphi'$: $\pi^*(x_i) = 1$ if $a_i = {\rm True}$ and $\pi^*(x_i) = 0$ if $a_i = {\rm False}$. $\pi^*(y_i) = 0$ if and only if  $C_i$ is satisfied by $\mathbf{a}$. Else $\pi^*(y_i) = 1/2$.   
For every satisfiable clause $C_i$, $\trust(C_i\vee y_i, \pi^*) = 1$ and $\trust(\neg y_i,\pi^*) = 1$. For all other clauses $C$ in $\varphi'$, $\trust(C,\pi^*) = 1/2$. Since there are $r$ clauses that are satisfied, the number of clauses for which $\trust(C, \pi^*) = 1/2$ is $2m-2r$. Hence the $\trust(\varphi',\pi^*) = \frac{1}{4^{(m-r)}}$.  Thus $\optrustval(\varphi') = \frac{1}{4^{m-r}}$. 
\end{proof}
Since $\optrustval(\varphi') = 1/4^{m-r}$,  $\maxsatval$ for $\varphi$ can be computed by knowing the $\optrustval$.
%This completes the proof of the theorem.
\end{proof}

While the above theorem shows that $\maxsatval$ can be computed from $\optrustval$, the next theorem shows that a maxsat assignment can be computed from an optimal $\mathbb{V}$-interpretation.
\begin{theorem}
$\maxsat$ metric reduces to $\optrust$.
\end{theorem}

\begin{proof}Consider the same reduction as from the previous theorem. Our task is to construct a \maxsat\ assignment for $\varphi$, given an optimal $\mathbb{V}$-interpretation $\pi$ for $\varphi'$. By the earlier theorem, $\trust(\varphi', \pi) = \frac{1}{4^{m-r}}$, where $r$ is the maximum number of satisfiable clauses of $\varphi$. 

We next establish a series of claims on the values takes by $\pi(y_i)$ and $\pi(x_i)$.

\begin{claim}\label{claim:zero-half}
For all $y_i$; $\pi(y_i) \in \{0,1/2\}$.
\end{claim}

\begin{proof}
Consider a clause $C_i' = (C_i \vee y_i)$ for which $\ell_{C'_i} = y_i$. Now the contribution of $C'_i$ and the clause $\neg{y_i}$ to $\trust(\varphi', \pi)$ is 
$ \pi(y_i) \times (1- \pi(y_i))$.
Since there is no clause $C'_j$ for which $\ell_{C'_j} = y_i$, the above value is maximized when $\pi(y_i) = 1/2$.
Now consider a clause $C'_j = (C_j \vee y_j)$, for which $\ell_{C'_j} \neq y_j$. Contribution of $C'_j$ and the clause $\neg{y_j}$ to $\trust(\varphi', \pi)$ is 
$\pi(\ell_{C'_j}) \times \pi(\neg{y_j)}$.
Since, $\ell_{C'_j} \neq y_j$, and there is no other clause in which $y_j$ or $\neg{y_j}$ appear, the above expression is maximized when $\pi(\neg{y_j}) = 1$ and thus $\pi(y_j) = 0$.
\end{proof}

\begin{claim}\label{clm:onlyzero}
For every $i$, if $y_i$ is not maximizing literal for clause $C'_i$, then $\pi(y_i) = 0$. 
\end{claim}

\begin{proof}
Let $C'_i$ be a clause for which $y_i$ is not maximizing literal. Say $\ell_j $ is the maximizing literal. We first consider the case $\pi(\ell_j) < 1/2$. By previous claim, $\pi(y_i) \in \{0, 1/2\}$, and if $\pi(y_i) = 1/2$, then $\ell_j$ can not be maximizing literal for clause $C'_i$. Thus $\pi(y_i) = 0$. Now consider the case $\pi(\ell_j) \geq 1/2$. Suppose that $\pi(y_i) = 1/2$. Now the contribution of the clauses $C'_i$ and $\neg{y_i}$ to $\trust(\varphi, \pi)$ is $\pi(\ell_j)/2$. However, if we change $\pi(y_i) = 0$, then the contribution of these clauses would become $\pi(\ell_j)$ and this would contradict the optimality of $\pi$. Thus  by Claim~\ref{claim:zero-half}, $\pi(y_i) = 0$.
\end{proof}

\begin{claim}\label{clm:zeroonehalf}
For all $x_i$, if $x_i$  or $\neg{x_i}$ is a maximizing literal, then  $\pi(x_i) \in \{0, 1, 1/2\}$
\end{claim}

\begin{proof}
We argue for the case when $x_i$ is a maximizing literal. The case when $\neg{x_i}$ is a maximizing literal follows by similar arguments.
Suppose that $x_i$ is a maximizing literal and $\pi(x_i) < 1/2$ and $\pi(x_i)$ is neither 0 nor 1. It must be the case that $\neg{x_i}$ is also a maximizing literal, otherwise making $\pi(x_i) =1 $ will increase the trust value. Suppose $x_i$ is a maximizing literal for $a$ many clauses and $\neg{x_i}$ is a maximizing literal for $b$ many clauses. If $a > b$, then we can obtain a $\mathbb{V}$-interpretation, by swapping  $\pi(x_i)$ with $\pi(\neg{x_i})$. If $a$ equals $b$, then $\pi(x_i)$ must be equal to $1/2$ as $x^a (1-x)^a$ is maximized for $x = 1/2$. Thus $a < b$. 
For every clause $C'_j$ for which $x_i$  or $\neg{x_j}$ is the maximizing literal, it must be the case that $\pi(y_j) = 0$, by  Claim~\ref{clm:onlyzero}. Let $\mathcal{C}$ be the collection of all clauses $C'_j$ together with $\neg{y_j}$, where either $x_i$ or $\neg{x_i}$ is maximizing literal.  The contribution of these clauses to $\trust(\varphi, \pi)$ is $\pi(x_i)^a \times (1- \pi(x_i))^b \times 1^{a+b}$.

We now construct a new $\mathbb{V}$-interpretation $\pi'$ that will contradict the optimality of $\pi$.
For every clause $C'_j \in \mathcal{C}$ in which $x_i$ is the maximizing literal, $\pi'(y_i) = 1/2$ and $\pi'(x_i) = 0$. Now the contribution of clauses from $\mathcal{C}$ to $\trust(\varphi, \pi')$ is
$(\frac{1}{2})^a \times 1^b \times (\frac{1}{2})^a \times 1^b$

Since $x^a(1-x)^b < 1/4^a$ (when $a < b$), 
\[(\frac{1}{2})^a \times 1^b \times (\frac{1}{2})^a \times 1^b > \pi(x_i)^a \times (1- \pi(x_i))^b \times 1^{a+b}\]
Thus $\trust(\varphi, \pi') > \trust(\varphi, \pi)$ which is a contradiction. Thus if $\pi(x_i) < 1/2$, then $\pi(x_i) = 0$, a similar argument shows that if $\pi(x_i) > 1/2$, then $\pi(x_i) = 1$.
\end{proof}

\begin{claim}\label{clm:equal}
For every $x_i$ with $\pi(x_i) = 1/2$, $x_i$ and $\neg{x_i}$ are maximizing literals for exactly the same number of clauses.
\end{claim}

\begin{proof}
Let $\mathcal{C}$ be the collection of clauses for which either $x_i$ or $\neg{x_i}$ is maximizing literal. Suppose that $x_i$ is maximizing literal for $a$ clauses and $\neg{x_i}$ is maximizing literal for $b$ clauses. If $a \neq b$, $\pi(x_i) = \frac{a}{a+b} \notin \{0, 1, 1/2\}$ and this contradicts Claim~\ref{clm:zeroonehalf}. 
\end{proof}

We will show how to construct a \maxsat\ assignment from $\pi$: If $\pi(x_i) = 0$, set  the truth value of $x_i$ to False, else set the truth value of $x_i$ to True. 

By Claim~\ref{clm:zeroonehalf}, $\pi(x_i) = \{0, 1/2, 1\}$. Let $H$ be the number of clauses for which  maximizing literal $\ell$ is a $x$-variable and $\pi(\ell) = 1$. Note that the above truth assignment will satisfy all the $H$ clauses.  Let $N$ be number of clauses for which maximizing literal $\ell$ is a $x$-variable and $\pi(\ell) = 1/2$. By Claim~\ref{clm:equal}, in exactly $N/2$ clauses a positive literal is maximizing, and thus all these $N/2$ clauses are satisfied by our truth assignment. Thus the total number of clauses satisfied by the truth assignment is $N/2+H$. Let $Y$ the number of clauses in which $y_i$ is maximizing literal. By Claim~\ref{claim:zero-half}, $\pi(y_i) = 1/2$ when $y_i$ is maximizing literal. Thus 
%\begin{eqnarray*}
\[\trust(\varphi', \pi) =   1^H \times (\frac{1}{2})^N \times (\frac{1}{2})^{2Y}
                  =  \frac{1}{4^{N/2+Y}}
                 =  \frac{1}{4^{m-r}}\]
%\end{eqnarray*}
The last equality follows from Claim~\ref{clm:earlierargument}. Thus $m-r = N/2+Y$, combining this with $m = H+N+Y$, we obtain that $N/2+H = r$. Thus the truth assignment constructed will satisfy $r$ clauses and is thus a $\maxsat$ assignment. 
 \end{proof}

\section{Approximating \optrustval}\label{sec:approxtrust}

We study the problem of approximating $\optrustval$ efficiently. Below, a $k$-SAT formula is a CNF formula with {\em exactly} $k$ distinct variables in any clause. We start with the following proposition.
%whose proof is omitted.

%appears in the appendix.

\begin{lemma}\label{lem:approx1}
Let $a_1, \cdots a_n$ be an assignment, that satisfies $r$ clauses of a CNF formula $\varphi(x_1, \cdots x_n)$. There is an interpretation $\pi$ so that $\trust(\varphi,\pi)$ is 
$\left(\frac{m-r}{m}\right)^{m-r} \left(\frac{r}{m}\right)^r$
\end{lemma}

%%%%%MOVED TO APPENDIX
\begin{proof}
If $a_i = 1$, set $\pi(x_i) = (1-\epsilon)$ and if $a_i =0$, then set $\pi(x_i) = \epsilon$. For every clause $C_i$ that is satisfied, we obtain a max value of $(1-\epsilon)$ and for every clause that is not satisfied, the max value is $\geq \epsilon$. Thus the \optrust\ obtained by this assignment is $(1-\epsilon)^r \epsilon^{m-r}$, and this is maximized when $\epsilon = \frac{m-r}{m}$ by Proposition~\ref{prop:basic}.
\end{proof}

Hence, for example, if $\varphi$ is a 3-SAT formula,  since a random assignment satisfies $7/8$  fraction of the clauses in expectation, for a random assignment $r \geq 7m/8$, and by Lemma \ref{lem:approx1}, $\optrustval(\varphi) > 0.686^m$. The following lemma shows that one can get a better lower bound on $\optrustval$ in terms of the clause sizes for CNF formulae.

\begin{lemma}\label{lem:approx2}
For every CNF formula $\varphi$, $\optrustval(\varphi) \geq e^{-\sum_i \frac{1}{k_i}}$ where $k_i$ is the arity of the $i$'th clause in $\varphi$.
\end{lemma}

%Proof appears in the appendix.
%%%%MOVED TO APPENDIX
\begin{proof}
Consider the interpretation $\pi$ that assigns every variable $x_i$ a uniformly chosen value in the interval $[0,1]$. Let the clauses in $\varphi$ be $C_1, \dots, C_m$.
Then:
\begin{align*}
\log \E[\trust(\varphi,\pi)] 
&\geq \E \log \trust(\varphi, \pi) ~(\mbox{Jensen's Inequality})\\
&= \sum_i \E\left[\log \max_{\ell \in C_i} \pi(\ell)\right]\\
&= -\sum_i \int_{-\infty}^0 \Pr\left[\log \max_{\ell \in C_i} \pi(\ell) \leq t\right] dt\\
&= -\sum_i \int_{-\infty}^0 \Pr\left[\max_{\ell \in C_i} \pi(\ell) \leq e^t\right] dt\\
&= -\sum_i \int_{-\infty}^0 e^{k_it} dt = -\sum_i \frac{1}{k_i}
\end{align*}
Hence, there exists a choice of $\pi$ achieving this trust value.
\end{proof}
This yields a probabilistic algorithm. For example, if $\varphi$ is a $3$-SAT formula, $\optrustval(\varphi) > 0.716^m$  and thus improving on the result of Lemma \ref{lem:approx1}. In fact, we can design a deterministic polynomial time algorithm that finds an interpretation achieving the trust value guaranteed by Lemma \ref{lem:approx2}, using the well-known `method of conditional expectation' to derandomize the construction in the proof (For example, see~\cite{AS08,GW94}). 

\begin{theorem}\label{thm:appalgo}
There is a polynomial-time, $e^{-m/k}$-approximation algorithm for $\optrust$, when the input formulas are $k$-CNF formulas with $m$-clauses.
\end{theorem}

%%%%MOVED TO APPENDIX
\begin{proof}

Arbitrarily ordering the variables $x_1, x_2, \dots, x_n$, the idea is to sequentially set $\pi^*(x_1), \pi^*(x_2), \dots, \pi^*(x_n)$ ensuring that for every $i$:
\[
\E_{\pi \leftarrow_U [0,1]^n}\left[ \log \trust(\varphi, \pi) \mid \pi(x_j) = \pi^*(x_j)~\forall j\leq i\right] \geq -\sum_i \frac{1}{k_i}. \tag{*}
\]
Assuming $\pi^*(x_1), \dots, \pi^*(x_{i-1})$ have already been fixed, we show how to choose $\pi^*(x_i)$ satisfying the above.  We use $\pi_{<i}$ to denote $\pi(x_1) \cdots \pi(x_{i-1})$. For a clause $C$, let $\alpha = \max_{\ell \in C \cap \{x_j, \bar{x}_j: j < i\}} \pi^*(\ell)$, and suppose $x_i \in C$. Then:
\begin{align*}
&= -\int_{-\infty}^0 \Pr_{\pi} \left[\log \max_{\ell \in C} \pi(\ell) \leq t \mid \pi_{<i} = \pi^*_{<i}, \pi^*(x_i) = p\right] dt\\
&= -\int_{\log \max(\alpha, p)}^0 \Pr_{\pi} \left[\log \max_{\ell \in C \cap \{x_j, \bar{x}_j : j>i\}} \pi(\ell) \leq t \right] dt\\
&= -\frac{1}{k'} \left(1 - \max(\alpha, p)^{k'}\right)
\end{align*}
where $k'$ is the number of literals in the clause  $C$ involving variables $x_{i+1}, \dots, x_n$. One can similarly evaluate the conditional expectation in the cases $\bar{x}_i \in C$ and $C \cap \{x_i, \bar{x}_i\} = \emptyset$. 

Summing up over all the clauses $C$, we get that \[\E_{\pi}\left[ \log \trust(\varphi, \pi) \mid \pi_{<i} = \pi^*_{<i}, \pi^*(x_i) = p\right]\] is a continuous function of $p$ that is a piecewise polynomial in at most $m$ intervals. In polynomial time\footnote{For simplicity, we ignore issues of precision here, but the error can be made inversely polynomial in $n$.}, we can find a value of $p$ that maximizes this function. By induction on $i$, the maximum value of this function is at least $-\sum_i \frac{1}{k_i}$, and hence (*) is satisfied. This completes the description of the algorithm. 

\end{proof}

%To summarize, this yields the following theorem.

Next, we show that the approximation factor $e^{-m/k}$ can not be significantly improved. 

We use the following result on hardness of approximating \maxsat\ established by Hastad~\cite{Hastad01}. 
\begin{theorem}[\cite{Hastad01}]\label{thm:hastad}
For any $\varepsilon > 0$ and any $k\geq 3$ it is NP-hard to distinguish satisfiable $k$-SAT formulas from $k$-SAT formulae $< m(1-2^{-k} +\varepsilon)$ satisfiable clauses.
\end{theorem}
We are now ready to show the following.

\begin{theorem}\label{thm:inapprox}
There is no polynomial-time ${1\over 4^{m(2^{-k}-\varepsilon)}}$-approximation algorithm for \optrust\ for $k$-SAT formulae, unless $\P=\NP$.    
\end{theorem}

\begin{proof}
Assuming such an approximation algorithm $A$ exists, we contradict Hastad's Theorem (Theorem \ref{thm:hastad}). 
Consider the following algorithm $A'$ that on input a $k$-SAT formula $\varphi$, runs $A(\varphi)$. If $A$ outputs a value  that is $\geq {1\over 4^{m(2^{-k}-\varepsilon)}}$, then $A'$ outputs YES otherwise outputs NO.  
Suppose $\varphi$ is satisfiable, then $\optrust(\varphi) = 1$. Hence $A$ will output a value $\geq {1\over 4^{m(2^{-k}-\varepsilon)}}$. Thus $A'$ output YES. 
Suppose maximum number of satisfiable clauses for $\varphi$ is $\leq m(1-2^{-k} +\varepsilon)$. By Theorem~\ref{thm:OptTrustBound}, 

%\begin{eqnarray*}
\[\optrust(\varphi) <  \frac{1}{4^{m-m(1-2^{-k}+\varepsilon)}} 
=  \frac{1}{4^{m(2^{-k}-\varepsilon)}} \]
%\end{eqnarray*}

Hence output of $A$ is $<{1\over 4^{m(2^{-k}-\varepsilon)}}$ and 
hence $A'$ will output NO. 

Thus $A'$ contradicts Theorem~\ref{thm:hastad}, unless $\P=\NP$. 
\end{proof}

Thus, for example for $3$-SAT formulas,  while we have a polynomial-time, $0.716^m$-approximation algorithm (by Theorem~\ref{thm:appalgo}), we cannot expect an efficient $0.845^m$-approximation algorithm by the above result unless $\P$ equals $\NP$.  It remains an interesting open problem to determine the optimal approximation ratio for this problem achievable by a polynomial time algorithm.

\section{Complexity of Access Maximization}\label{sec:access}

In this section, we study the optimization problems for the access control semiring  $\mathbb{A}_k = ([k], \max, \min, 0, k)$. We refer to the corresponding computational problems as \optaccessval\ and \optaccess. For this section we first assume the negation function is the additive inverse modulo $k$. That is $\negation(a) = b$ such that $a+b \equiv 0 ~({\rm mod}~k)$. 
%This can be relaxed to a general negation function to get identical complexity results. 

\begin{theorem}\label{thm:access}
Let $\varphi(x_1, \cdots x_n)$ be a propositional formula in negation normal form and $\mathbb{A}_k = ([k], \max, \min, 0, k)$. The following statement holds. 
\begin{itemize}
\item If $\varphi$ is satisfiable, then $\optaccessval(\varphi) = k$.
\item If $\varphi$ is not satisfiable, then $\optaccessval(\varphi) =\lfloor\frac{k}{2}\rfloor$.
\end{itemize}
\end{theorem}

\begin{proof}
We will first prove it for the case when  $\varphi$ is in the CNF form, i.e  $\varphi = C_1 \wedge \cdots \wedge C_m$. Suppose that the formula is satisfiable and $a_1 \cdots a_n$ is a satisfying assignment to the variables $x_1, x_2, \cdots, x_n$.  Consider the interpretation $\pi$ defined as follows: If $a_i$ is true, then $\pi(x_i) = k$, else $\pi(x_i) = \negation(k)$. 
Consider a clause $C$, since the formula is satisfiable, there exists a literal $\ell_i$ (either $x_i$ or $\neg x_i$ for some $i$)  in $C$ such that $\ell_i$ is set to true. If $\ell_i = x_i$, then $\pi(x_i) = k$ and $\sem(x_i, \pi) = k$. If $\ell_i = \neg x_i$, then $\pi((x_i) = \negation(k) = 0$ and $\sem(\neg x_i, \pi) = \negation(0) = k$. Since $C$ is a disjunction $\sem(C, \pi) = k$. Thus for every clause $C_i$, $\sem(C_i, \pi) = k$.
Since $\varphi$ is a conjunction of $C_1, \cdots C_m$, it follows that $\sem(\varphi, \pi) = k$.

For the proof of the second item, first assume that $k$ is even, the proof when $k$ is odd is very similar.  Note that in this case, $\negation(k/2) = k/2$. Let $\varphi = C_1 \wedge \cdots \wedge C_m$ be an unsatisfiable formula. Consider an interpretation $\pi$ where $\pi(x_i) = k/2$ for every $1 \leq i \leq n$. Clearly, for this interpretation, $\sem(\varphi, \pi) = k/2$.   Suppose that $\pi'$ be an interpretation  $\sem(\varphi, \pi') > k/2$. Consider the following satisfying assignment: $a_i$ is true if $\varphi'(x_i) > k/2$, else $a_i$ is false. Observe that this is a consistent assignment. We will establish that this assignment satisfies $\varphi$. This establishes that $\optaccessval(\varphi) = k/2$.

Note that for every clause $C_j$, $1 \leq j \leq m$, $\sem(C_j, \pi') > k/2$. Fix a clause $C$, since $\sem(C,  \pi') > k/2$, there exists a literal $\ell_i$ in $C$ such that $\sem(\ell_i, \pi') > k/2$. If $\ell_i = x_i$, then $\sem(x_i, \pi') > k/2$ which implies that $\pi'(x_i) > k/2$. Thus $a_i$ is true and the clause $C$ is satisfied by the assignment. If $\ell_i = \neg x_i$, then $\sem(\neg x_i, \pi') > k/2$. Thus $\negation(\pi'(x_i)) > k/2$. By the definition of $\negation$, we have $\pi'(x_i) < k/2$. Thus $a_i $ is set to false. Thus the clause $C$ is satisfiable. This proves that the assignment $a_1, \cdots, a_n$ satisfies the formula $\varphi(x_1, \cdots, x_n)$.

The case where the general formula is in the negation normal form follows by similar ideas using the notion of proof trees as in the case of Viterbi semiring. 
\end{proof}

For a general negation function, we can establish an analogous theorem. For this, we define the notion of the {\em index of negation}. Given a negation function $\negation$, its {index} denoted by ${\it Index}(\negation)$ is the largest $\ell$ for which there exists $a \in [k]$, such that both $a$ and $\negation(a)$ are at least $\ell$. 

\begin{theorem}
Let $\varphi(x_1, \cdots x_n)$ be a propositional formula in negation normal form and $\mathbb{A}_k = ([k], \max, \min, 0, k)$. The following statement holds. 
\begin{itemize}
\item If $\varphi$ is satisfiable, then $\optaccessval(\varphi) = k$.
\item If $\varphi$ is not satisfiable, then $\optaccessval(\varphi) = Index(\negation)$.
\end{itemize}
\end{theorem}

The following is a corollary to the above result and its proof which states that the complexity of optimization problems over access control semiring is equivalent to their complexity over the Boolean semiring. 

\begin{theorem}\label{thm:accessequivalence}
The problem $\optaccessval$ and ${\mathsf{SAT}}$ are equivalent under metric reductions. 
Similarly, the problem \optaccess\ and the problem of computing a satisfying assignment of a given Boolean formula are equivalent under metric reductions.  
\end{theorem}

\section{Conclusion}\label{sec:conclusion}

In this work, we provided a comprehensive study of the computational complexity of $\optsem$ and the related problem \optsemval\  over various semirings such as Viterbi semiring, tropical semiring, access control semiring and fuzzy semiring,  from both an algorithmic and a complexity-theoretic viewpoint. An exciting recent development in the field of CSP/SAT solving has been the development of solvers for {\lexsat}, which seeks to find the smallest lexicographic satisfying assignment of a formula~\cite{MAL11}. In this regard, Theorem~\ref{theorem:pnp} opens up exciting directions of future work to develop efficient techniques for {\optrust}.

\section{Acknowledgements}
We thank Val Tannen for introducing us to the world of semiring semantics and for helpful conversations during the nascent stages of the project.  We thank the anonymous reviewers of AAAI-23 for valuable comments. This research is supported by the National Research Foundation under the NRF Fellowship Programme [NRF-NRFFAI1-2019-0004] and Campus for Research Excellence and Technological Enterprise (CREATE) program. Bhattacharyya was supported in part by the NRF Fellowship Programme [NRF-NRFFAI1-2019-0002] and an Amazon Research Award.  Vinod was supported in part by NSF CCF-2130608 and NSF HDR:TRIPODS-1934884 awards. Pavan was supported in part by NSF CCF-2130536,  and NSF HDR:TRIPODS-1934884 awards.

\bibliography{main}
\bibliographystyle{alpha}

\end{document}